\documentclass[12pt,a4paper]{article}
\usepackage[margin=2.75cm]{geometry}

\usepackage{latexsym}
\usepackage{amssymb,amsmath,amsthm}
\usepackage{nicefrac}
\usepackage{url}
\usepackage[hidelinks]{hyperref}

\usepackage[round]{natbib}
\renewcommand{\paragraph}[1]{\textbf{#1}}

\title{\bf On the Difficulty of Measuring Divisiveness of Proposals under Ranked Preferences}

\author{Ulle Endriss \\[3pt] 
Institute for Logic, Language and Computation \\ University of Amsterdam}

\newtheorem{theorem}{Theorem}
\newtheorem{proposition}[theorem]{Proposition}

\newtheorem{definition}{Definition}
\newtheorem{axiom}{Axiom}

\theoremstyle{definition}

\newtheorem{examplex}{Example}
\newenvironment{example}
  {%
   \pushQED{\qed}\begin{examplex}}
  {\popQED\qedhere\end{examplex}}
 
\newcommand{\Profs}{{X!}^{N\subset\mathbb{N}}}
\newcommand{\pref}{\succ}
\newcommand{\pos}{\textit{pos}}
\renewcommand{\complement}[1]{\overline{#1}}

\DeclareMathOperator*{\argmax}{argmax}
\DeclareMathOperator*{\argmin}{argmin}

\renewcommand{\geq}{\geqslant}

\begin{document}
\maketitle

\begin{abstract}\noindent
Given the stated preferences of several people over a number of proposals regarding public policy initiatives, some of those proposals might be judged to be more ``divisive'' than others. When designing online participatory platforms to support digital democracy initiatives enabling citizens to deliberate over such proposals, we might wish to equip those platforms with the functionality to retrieve the most divisive proposals currently under discussion. Such a service would be useful for analysing the progress of deliberation and steering discussion towards issues that still require further debate. Guided by this use case, we explore possibilities for providing a clear definition of what it means to select a set of most divisive proposals on the basis of people's stated preferences over proposals. Then, employing the axiomatic method familiar from social choice theory, we show that the task of selecting the most divisive proposals in a manner that satisfies certain seemingly mild normative requirements faces a number of fundamental difficulties.
\end{abstract}

\section{Introduction}

The term \emph{digital democracy} refers to the idea of attempting to seize the opportunities provided by modern information technologies, and specifically the internet, to make democratic decision making more participatory and flexible \citep{HackerVanDijk2000,HelbingEtAL2023}. Yet, successful examples for realising this potential are still few and far between. Besides obvious challenges of usability, there are even more intricate and pressing challenges of legitimacy~\citep{Fung2015}. In a recent white paper, \citet{GrossiEtAl2024} argue that overcoming these challenges requires not just technological ingenuity but also scientific rigour, including in particular the careful development of formal and computational models of digitally-enhanced civic participation. 

In this spirit, we put forward a formal model for analysing the \emph{divisiveness} of a proposal under discussion on a digital democracy platform on the basis of the preferences reported by its users. While most people will have some intuitive understanding of what being ``divisive'' might mean, it is not at all obvious how one should map those intuitions into a mathematically precise definition. To appreciate the relevance of our objective of arriving at such a definition, imagine a platform where citizens can deliberate over, vote on, or otherwise express their preferences regarding a number of proposals concerning public policy. Depending on context, proposals might range from matters of national and supranational importance, such as reforming carbon taxing, to inherently local issues, such as the infamous pedestrianisation of Norwich city centre. Such a platform might offer the service of identifying proposals that are especially divisive in view of the preferences reported so far. This can be useful in a number of ways:
$(i)$~to select proposals on which to encourage further debate,
$(ii)$~to explain which issues prevent a population from reaching consensus~\citep{NavarreteEtAlNHB2024}, and
$(iii)$~to detect instances of ``political posturing'' by individuals excessively focusing on divisive rather than common-value issues~\citep{AshEtAlJP2017}.

Our objective of providing formal foundations for the process of extracting divisive proposals from the reported preferences of a given population in a principled manner is inspired by recent work by \citet{NavarreteEtAlNHB2024}. In two online experiments they ran during the French and Brazilian presidential election campaigns in 2022, they collected a large amount of data on citizen preferences regarding policy proposals featuring in the election manifestos of the political parties involved. Central to their analysis of the data collected is the idea that focusing on proposals that are particularly divisive (rather than, say, particularly popular) can lead to deep insights regarding the wishes and needs of citizens. To this end, \citeauthor{NavarreteEtAlNHB2024} propose one specific approach to measuring divisiveness. That approach is intuitively appealing and it enabled them to perform an insightful analysis of their data. But there is no reason to believe that it is the only valid (or the best) approach. 

To answer the question of what might be the \emph{right} way of measuring divisiveness we propose to turn to social choice theory~\citep{ArrowEtAlHBSCW2002}, the study of mechanisms for collective decision making. Especially computational social choice~\citep{HBCOMSOC2016}, with its focus on methods from Computer Science and AI, appears naturally suited to this task.

To make our objective of defining divisiveness concrete, we focus on functions that map any given profile of individual preferences (i.e., strict rankings of the available proposals) to a subset of the set of all proposals---to be interpreted as the most divisive proposals. Indeed, the service described earlier would need to implement precisely such a function. Formally, it is of the same ``type'' as a voting rule~\citep{ZwickerHBCOMSOC2016}, which also takes as input a preference profile and returns as output a set of the items the voters report preferences on, namely the set of election winners (due to ties there could be more than one). But our interpretation of these functions is very different. Clearly, the ``best'' proposals (returned by a voting rule) are not, usually, also the ``most divisive'' proposals. So we can use the \emph{methodology} of social choice theory to study these new functions, but we cannot rely on existing \emph{results} regarding voting rules.

\paragraph{Related work.}
As mentioned earlier, the idea of attempting to measure the divisiveness of proposals in view of a population reporting preferences over those proposals originates with \citet{NavarreteEtAlNHB2024}. 
In follow-up work, \citet{ColleyEtAlIJCAI2023} generalise the specific measure of divisiveness put forward by \citeauthor{NavarreteEtAlNHB2024} to a broader class of measures and initiate the study of some of the formal properties of such measures. Most relevant to our interests here is their analysis of how the measures they consider perform on specific ``extreme'' preference profiles, corresponding to populations that are completely polarised or in full agreement on certain issues.

The work of \citet{DelemazureEtAlIJCAI2024} is in a similar spirit. They design and analyse rules to select \emph{pairs of proposals} that, in some sense, maximise disagreement regarding the relative ranking of these proposals. 

We stress that here we are interested in divisiveness only \emph{as a function of reported preferences}. This is different, for instance, from the work of \citet{AshEtAlJP2017}, who in a study of US congressional speech records label proposals as ``divisive'' when they include tuples of words such as $\{\textit{right},\textit{bear},\textit{arm}\}$ or $\{\textit{tax},\textit{break},\textit{wealthy}\}$. 

While the idea of measuring the divisiveness of individual proposals within a preference profile is fairly new, there is a somewhat more established (but still not at all conclusive) literature on the closely related question of how to measure levels of \emph{polarisation} or \emph{diversity} of a given preference profile in its entirety. This is relevant to our problem, as we should expect to find a positive correlation between the number of highly divisive proposals on the one hand and the degree of polarisation or diversity in a population on the other. The literature on polarisation and diversity spans multiple disciplines, and formal work of the kind we set out to do here can be found in both Economics and AI. 
For instance, \citet{AlcaldeUnzuVorsatzSCW2013} axiomatise several measures of \emph{cohesiveness} for preference profiles---a notion that in some sense is dual to polarisation and diversity, while \citet{CanEtAlMSS2015} directly axiomatise a measure of polarisation of preference profiles.
\citet{HashemiEndrissECAI2014} define the problem of measuring the diversity of preference profiles, suggest several indices for doing so as well as axioms one might want to use to analyse those indices, and then show experimentally how the performance of different voting rules can be affected by the diversity of the preference profile at hand.
Finally, \citet{FaliszewskiEtAlIJCAI2023} critically examine a number of earlier proposals for measuring diversity, polarisation, and similar notions and put forward terminology to adequately describe such phenomena. They do so with a view towards running experimental studies in computational social choice---where we need good terminology to describe the data and want that data to be representative.

In methodological terms, we will largely rely on the \emph{axiomatic method} of social choice theory and related fields~\citep{ArrowEtAlHBSCW2002,HBCOMSOC2016,Thomson2023}.

\paragraph{Contribution.}
We put forward a formal model of divisiveness by introducing the concept of \emph{divisiveness selection function}, mapping any given profile of ranked preferences over proposals to the set of those proposals that are most divisive in view of that profile.
We then suggest three approaches for designing such functions, each of which seeks to base the definition of divisiveness on that of a related concept for which convincing formalisations are already available in the literature:
$(i)$ \emph{scoring functions}, $(ii)$~\emph{voting rules} (or \emph{social choice functions}), and $(iii)$~measures for profile properties such as \emph{polarisation} or \emph{diversity}. 
The first approach generalises but still is closely inspired by the work of \citet{NavarreteEtAlNHB2024} and \citet{ColleyEtAlIJCAI2023}; the other two approaches are new.
Finally, we formulate several normative desiderata for divisiveness selection functions, we illustrate how to use them to help us identify a suitable notion of divisiveness for a given application, and we prove that certain desiderata are inherently incompatible with one another.
Together, our results illustrate the difficulty of finding a mathematically precise and normatively adequate definition of ``divisiveness''.

\paragraph{Paper overview.}
We present our model of divisiveness selection functions in Section~\ref{sec:model}, and our three approaches for defining such functions in Section~\ref{sec:classes}. We then propose a number of axioms (i.e., normative desiderata) for such functions in Section~\ref{sec:axioms}, and illustrate how they can help us orient ourselves within the huge space of possible definitions of divisiveness. We finally show, in Section~\ref{sec:impossibilities}, how certain combinations of desiderata lead to impossibility results.

\section{The Model}\label{sec:model}

In this section we put forward our formal model of divisiveness. The central concept we introduce is that of a \emph{divisiveness selection function}, mapping the preferences over policy proposals reported by the members of a given population to the set of those proposals that should be deemed most divisive in view of these preferences. Within this model it will then be possible to define and analyse many different concrete proposals for how to measure divisiveness.

\paragraph{Notation and terminology.}
Let $X$ be a finite set of \emph{proposals}, and let $m=|X|$ be the number of those proposals. Suppose a finite number of \emph{agents} express their \emph{preferences} by each reporting a ranking of the proposals in $X$. Let ${X!}$ be the set of all strict linear orders on~$X$, i.e., the set of all possible preferences. We represent agents as natural numbers $i\in\mathbb{N}$, so the \emph{electorate} (or \emph{population}) reporting preferences in any given situation is a finite nonempty set $N \subset \mathbb{N}$. When each agent~$i$ belonging to some population $N$ reports a preference, we obtain a \emph{profile} of preferences. 
We model profiles as functions $R : N \to {X!}$, mapping each agent to her preference. For notational convenience, we often write $R_i$ for the preference $R(i)$ of agent~$i$ in profile~$R$. We denote the set of all conceivable profiles as $\Profs$ (with $N$ always understood to be nonempty and finite).

\paragraph{Selecting maximally divisive proposals.}
Given the information contained in a profile~$R$, we would like to identify the most divisive proposals in~$X$ (possibly just one). Our main objective here is to make this notion of being ``most divisive'' mathematically precise. So we are looking for a mathematical object of the following kind.

\begin{definition}
A \textbf{divisiveness selection function} (DSF) is a function $\Delta : \Profs \to 2^X \setminus \{\emptyset\}$, mapping any given profile of preferences over proposals to a nonempty set of proposals---to be interpreted as the set of most divisive proposals given those preferences. 
\end{definition}

\noindent
Observe that---as a mathematical object---a DSF is of the same type as a voting rule \citep{ArrowEtAlHBSCW2002,ZwickerHBCOMSOC2016}: both take a profile of preferences as input and return a nonempty set of alternatives as output. But the intended semantics of these sets is very different, so we should not assume that any function of this type that makes for a good voting rule would be of any interest as a DSF.

\noindent\textbf{Additional notation.}
Let us fix some further pieces of notation. 
First, we write $x\pref^R_i y$ in case agent~$i$ ranks proposal~$x$ above proposal~$y$ in profile~$R$. 
Second, we use $\pos^R_i(x)$ to refer to the \emph{position} of proposal~$x$ in the ranking reported by agent~$i$ in profile~$R$ (with the convention that the top position is referred to as position~$1$):
\begin{eqnarray*}
\pos_i^R(x) & = & 1 + \#\{\, y\in X \mid y \pref^R_i x \,\}
\end{eqnarray*}
The next three pieces of notation all concern sets $C\subseteq N$ (for some given electorate~$N$) we might think of as \emph{coalitions}, or \emph{subelectorates}, or \emph{subpopulations}.
First, for any coalition~$C$ we use $\smash{\complement{C}}$ to refer to its complement $N\setminus C$, when the relevant electorate~$N$ is clear from context.
Second, $N^R_{x\pref y}$ denotes the coalition of agents who rank $x$ above $y$ in profile~$R$.
Third, $R{\restriction_C}$ is the restriction of profile $R : N \to {X!}$ to coalition $C\subseteq N$; i.e., $R{\restriction_C}(i)=R(i)$ for each $i\in C$.
For profiles $R : N \to {X!}$ and $R' : N' \to {X!}$ with disjoint electorates ($N\cap N'=\emptyset$), we use $R\oplus R'$ to refer to the \emph{union} of $R$ and $R'$ (i.e., the function behaving like $R$ on $N$ and like $R'$ on~$N'$).

\paragraph{Examples.}
Let us now review two examples for how one might want to define a DSF. 
A first idea is to extract for each proposal the list of positions it occurs in within the profile at hand; to think of each such list as having been generated by a random variable; and to then equate the divisiveness of a proposal with the variance of the corresponding list of positions. \citet{ColleyEtAlIJCAI2023} call this approach to measuring divisiveness \emph{rank variance} and trace its origins back to work on measuring polarisation in social networks~\citep{MuscoEtAlWWW2018}. 
While \citeauthor{ColleyEtAlIJCAI2023} define what we earlier called an index function, here we turn that definition into a definition of a~DSF.

\begin{definition}\label{def:rankvar}
Under the \textbf{rank-variance DSF} $\Delta_{\textit{var}}$ the proposals deemed most divisive are those with the highest variance across the positions at which they occur in a given profile $R : N \to {X!}$:
\begin{eqnarray*}
\Delta_{\textit{var}}(R) & = & \argmax_{x\in X} \frac{1}{|N|} \cdot \sum_{i\in N} (\pos^R_i(x) - \mu^R(x))^2 \\
&& \text{where}\
\mu^R(x) = \frac{1}{|N|}\cdot\sum_{i\in N} \pos^R_i(x)
\end{eqnarray*}
\end{definition}

\noindent
We stress that \citet{ColleyEtAlIJCAI2023} do not specifically advocate this as a suitable definition for divisiveness. Indeed, while it seems natural to assume that greater disagreement about the positioning of a proposal contributes to greater divisiveness of that proposal, it is at best unclear whether the manner in which we measure the variance of a random variable also happens to be appropriate for measuring divisiveness. We are going to return to this issue in Section~\ref{sec:axioms}.

Inspired by seminal work by \citet{EstebanRayEcon1994} on polarisation,
\citet{NavarreteEtAlNHB2024} make an alternative---and rather intriguing---proposal for how to measure divisiveness. They argue that a proposal is divisive if two complementary subpopulations are likely to disagree on its merits. To operationalise this idea, they make specific assumptions about what the relevant subpopulations are and how we should measure disagreement. 
For this measurement, they use a scoring function, mapping each proposal~$x$ in a profile~$R$ to a real-valued score. Let us look at an example.
Recall that the \emph{Borda score} of $x$ in profile~$R$ can be computed by summing, across all preferences~$R_i$ occurring in~$R$, the number of proposals dispreferred to $x$ in $R_i$ \citep{ZwickerHBCOMSOC2016}. The \emph{normalised Borda score} then is the Borda score divided by the number of preferences in~$R$. Normalisation is useful if we need to compare scores across profiles of different size. 
We are now ready to adapt the definition of \citeauthor{NavarreteEtAlNHB2024} to our setting.

\begin{definition}\label{def:navarrete}
The \textbf{DSF of \citeauthor{NavarreteEtAlNHB2024}} $\Delta_{\textit{nav}}^s$ induced by scoring function~$s$ (mapping proposals in profiles to reals) is defined as: 
\begin{eqnarray*}
\Delta_{\textit{nav}}^s(R) & = & \argmax_{x\in X} \frac{1}{m-1} \cdot \sum_{y\not=x} \mathit{div}(R,x,y), 
\end{eqnarray*}
where $\mathit{div}(R,x,y) = |s(R{\restriction_{N^R_{x\pref y}}},x) - s(R{\restriction_{N^R_{y\pref x}}},x)|$, unless $N^R_{x\pref y}$ or its complement are empty, in which case $\mathit{div}(R,x,y) = 0$.
\end{definition}

\noindent
As here we are only interested in maxima, the normalisation factor $\nicefrac{1}{m-1}$ is technically redundant.
For their data analysis, \citeauthor{NavarreteEtAlNHB2024} focus on the case where $s$ returns the normalised Borda score of a given proposal in a given (partial) profile.
Thus, to determine the divisiveness of $x$ in $R$, for every proposal $y$ other than $x$, we split the electorate into two subelectorates, those that rank $x$ above $y$ and those that do not, and then compute the difference between the normalised Borda scores of $x$ for those two subelectorates. The divisiveness of $x$ then is the average such difference.\footnote{Recall that we do not have access to demographic information regarding participants. If we did, we could also identify proposals that are maximally divisive across, say, the urban/rural divide or the old/young divide. Here this is not possible. We only have access to the reported preferences and we must define divisiveness in those terms.}
\citeauthor{NavarreteEtAlNHB2024} demonstrate that, on their datasets of crowdsourced preferences regarding policy proposals extracted from party manifestos in France and Brazil, this definition of divisiveness can support an insightful analysis of complex preference data in practice.

\paragraph{Why selection?}
We conclude this section with a brief remark on methodology.
Framing the problem of defining divisiveness as one of defining a \emph{selection function} returning the most divisive proposals is but one of multiple natural approaches one could take. Another approach would be to define an \emph{index function} that maps each proposal to a cardinal value reflecting its degree of divisiveness. This is the approach followed by \citet{NavarreteEtAlNHB2024} and \citet{ColleyEtAlIJCAI2023}. 
It is very helpful when presenting statistics and visualisations, but it can be hard to justify the choice of a very specific cardinal divisiveness value for a given proposal. Yet another approach would be to look for a \emph{ranking function} that returns a (strict or weak) ranking of the proposals, but without attaching cardinal values to proposals.  
Here we focus on selecting maximally divisive proposals rather than ranking proposals in terms of their divisiveness or indexing them with degrees of divisiveness, because this is the most fundamental problem of the three. Indeed, a satisfactory solution to either one of the other two problems will immediately yield also a solution to the selection problem we focus on here. Thus, any negative results we might obtain for the selection problem will immediately carry over to both the ranking and the indexing problem as well.

\section{Classes of Divisiveness Selection Functions}\label{sec:classes}

So far we have seen two concrete proposals for the definition of a DSF. To map out the---enormous---terrain of possible definitions, in this section we propose three ``recipes'' for defining a DSF, each of which gives rise to a class of functions one might want to use.

\paragraph{Score-based definitions.}
\citet{ColleyEtAlIJCAI2023} generalise the proposal of \citet{NavarreteEtAlNHB2024} (see Definition~\ref{def:navarrete}) to a family of functions. We now generalise further. The idea is that we imagine that the full electorate $N$ will split into two complementary coalitions $C$ and~$\smash{\complement{C}}$, and each possible split will occur with a certain probability. To determine the divisiveness of $x$ we calculate the \emph{expected difference in score} in $C$ and $\smash{\complement{C}}$, under a scoring function of our choice.

Let $s$ be a scoring function mapping pairs of a profile restriction $R{\restriction_C}$ and a proposal~$x$ to the reals. (The aforementioned normalised Borda score is a good example.) We define the \emph{$s$-divisiveness} of proposal~$x$ in profile $R : N \to {X!}$ relative to the decomposition $\{C,\complement{C}\}$ of $N$ into two complementary subelectorates as follows:
\begin{eqnarray*}
\mathit{div}_s(R,x,C,\complement{C}) & \!\!=\!\! & \left\{
\begin{array}{ll}
\!|s(R{\restriction_C},x) - s(R{\restriction_{\complement{C}}},x)| & \text{if}\ C,\complement{C} \not= \emptyset \\
\!0 & \text{otherwise}
\end{array}\right.
\end{eqnarray*}
That is, we compute the difference in score for $x$ for the two subprofiles (but return~$0$ in case one of the subelectorates is empty, i.e., in case there is no well-defined score).

For any profile $R : N \to {X!}$ and proposal~$x$, let $p_{R,x}$ be a probability distribution on decompositions $\{C,\complement{C}\}$ of $N$ into complementary subelectorates. (The uniform probability distribution is a good example.) We now obtain a definition of divisiveness of $x$ if we imagine we are sampling possible decompositions using $p_{R,x}$ and then examine the divisiveness of $x$ relative to those decompositions.

\begin{definition}\label{def:scoreDSF}
The \textbf{score-based DSF} $\Delta^p_s$ induced by scoring function~$s$ and array of probability distributions~$p$ is defined as follows:
\begin{eqnarray*}
\Delta^p_s(R) & \!=\! & \argmax_{x\in X} \sum_{\{C,\complement{C}\}} p_{R,x}(\{C,\complement{C}\}) \cdot \mathit{div}_s(R,x,C,\complement{C})
\end{eqnarray*}
\end{definition}

\noindent
We stress that this definition is very general, and we may not always need or want this level of generality. For instance, it seems reasonable to assume that the probability $p_{R,x}(\{C,\complement{C}\})$ of sampling $\{C,\complement{C}\}$ does in fact not depend on either $R$ or $x$ (e.g., when every possible split is equally likely to occur). But note that the dependency of $p$ on both $R$ and $x$ is required to be able to cover Definition~\ref{def:navarrete}.\footnote{The probability distributions $p_{R,x}$ that allow us to obtain Definition~\ref{def:navarrete} as a special case of Definition~\ref{def:scoreDSF} are such that \smash{$p_{R,x}(\{C,\complement{C}\}) = 0$} unless there exists a $y$ such that $C$ or \smash{$\complement{C}$} coincide with \smash{$N^R_{x\pref y}$}. When considered from this angle, the specific choices made for Definition~\ref{def:navarrete} may seem somewhat arbitrary. (Why should we only sample these specific coalitions?) But there is an important practical advantage to this design choice: the number of coalitions now is polynomial rather than exponential.} 

Now, why should we choose one scoring function~$s$ over another? Natural candidates for $s$ are closely linked to well-studied voting rules. For instance, the (normalised) Borda score defines the \emph{Borda rule}~\citep{ZwickerHBCOMSOC2016}, which is widely used in practice, e.g., to elect officials for a range of professional organisations, and which comes with strong normative foundations, as it has been characterised as being the only voting rule that satisfies certain combinations of normatively appealing axioms~\citep{ChebotarevShamisAOR1998,YoungJET1974}. It is then tempting to think that we might be able to use existing axiomatic results for a voting rule induced by $s$ to also justify the use of $\Delta^p_s$ when analysing divisiveness. But this would be a fallacy, as the following example~shows. 

\begin{example}\label{ex:fallacy}
The \emph{positional scoring rules}~\citep{ZwickerHBCOMSOC2016} are voting rules that are induced by a scoring vector $(s_1,\ldots,s_{m}) \in \mathbb{R}^m$ by stipulating that a proposal garners $s_k$ points whenever an agent ranks it at position~$k$. For instance, for $m=4$ the Borda rule is induced by the scoring vector $(3,2,1,0)$. 
But for any fixed electorate size $|N|$ and any sufficiently small $\epsilon>0$ (with $|N|\cdot\epsilon < 1$), the Borda rule is also induced by $(3,2,1,\epsilon)$. Now consider this profile:
\[\begin{array}{c}
a\pref_1 b\pref_1 c\pref_1 d \\
b\pref_2 a\pref_2 c\pref_2 d \\
a\pref_3 b\pref_3 d\pref_3 c \\
b\pref_4 a\pref_4 d\pref_4 c 
\end{array}\]
For the DSF of Definition~\ref{def:navarrete}, which is an instance of the class of DSFs of Definition~\ref{def:scoreDSF}, when scoring according to $(3,2,1,0)$ we find all four proposals to be maximally divisive, while when scoring according to $(3,2,1,\epsilon)$, only $a$ and $b$ are maximally divisive.

Thus, as both of these scoring vectors induce the Borda rule, normative justifications of the Borda rule cannot possibly serve as a justification for preferring Definition~\ref{def:navarrete} with Borda scoring over its variant where we score proposals in the lowest position with~$\epsilon$ rather than~$0$.
\end{example}

\noindent
In a nutshell, we have argued that $(i)$~scoring function~$s$ inducing voting rule~$F$ and $(ii)$~the availability of a convincing normative argument for the use of $F$ together do not, in and of themselves, provide us with a convincing normative argument for using a score-based DSF $\Delta^p_s$ induced by $s$. This is so, because $s$ might not \emph{uniquely} induce $F$; instead, there typically are many other scoring functions~$s'$ that also induce $F$ but that would induce a different~DSF. 
Of course, this does not invalidate the approach of defining a DSF in terms of a scoring function, but it means that we have to justify the choice of scoring function from scratch.

\paragraph{SCF-based definitions.}
Our next class of DSFs is designed to address this concern. Rather than defining a DSF in terms of a scoring function we might try to define it directly in terms of a voting rule. Recall that a \emph{voting rule}, also known as a \emph{social choice function} (SCF), is a function $F$ mapping profiles (of any size) to a nonempty set of proposals---namely those that are, in some sense, the collectively most preferred proposals. A host of voting rules have been proposed and studied in the literature~\citep{ZwickerHBCOMSOC2016}; the Borda rule is one of many examples.

Given a manner of sampling pairs of complementary coalitions and a voting rule~$F$, we now want to model divisiveness of $x$ as the probability of encountering two coalitions that, under $F$, disagree on the election of~$x$. We shall require some further notation:\footnote{In case one of 
$C$ and $\smash{\complement{C}}$ should be empty, 
we define \smash{$\mathit{div}_F(R,x,C,\complement{C})$} as $0$.}
%
%
\begin{eqnarray*}
\mathit{div}_F(R,x,C,\complement{C}) & = & \left|\, \frac{[x\in F(R{\restriction_C})]}{|F(R{\restriction_C})|} -  \frac{[x\in F(R{\restriction_{\complement{C}}})]}{|F(R{\restriction_{\complement{C}}})|} \,\right|
\end{eqnarray*}
Here $[\cdot]$ is the Iverson bracket. Note that an expression of the form $\nicefrac{[x\in S]}{|S|}$ is equal  $\nicefrac{1}{|S|}$ if $x\in S$ and equal to $0$ otherwise.
So in case $F$ happens to return singletons for the relevant profiles, $\mathit{div}_F(R,x,C,\complement{C})$ reduces to $1$ when $F$ returns distinct winners for the two subprofiles and to $0$ otherwise. In the general case, where $F$ might return a larger set, we can think of $\mathit{div}_F(R,x,C,\complement{C})$ as the difference between the two probabilities of the two subelectorates to elect $x$ under $F$ in combination with a uniform tie-breaking rule.
This new notion of divisiveness of $x$ relative to two given subprofiles now gives rise to the definition of a new class of DSFs.

\begin{definition}\label{def:scfDSF}
The \textbf{SCF-based DSF} $\Delta_F^p$ induced by SCF $F$ and array of probability distributions~$p$ is defined as follows:
\begin{eqnarray*}
\Delta^p_F(R) & = & \argmax_{x\in X} \sum_{\{C,\complement{C}\}} p_{R,x}(\{C,\complement{C}\}) \cdot \mathit{div}_F(R,x,C,\complement{C})
\end{eqnarray*}
\end{definition}

\noindent
This definition can be instantiated with any SCF~$F$. 
It naturally captures the idea of divisiveness as the probability of two complementary subelectorates electing a different most preferred proposal under a given voting rule. Importantly though, as the following example demonstrates,\footnote{This example was suggested to me by an anonymous reviewer of an earlier version of the paper.} this can result in counterintuitive choices for certain profiles.

\begin{example}\label{ex:scf-unanimous-indecisive}
Consider any profile in which half of the agents report the first preference and the other half report the second one below:
\[\begin{array}{c}
a\pref b\pref c \\
a\pref c\pref b 
\end{array}\]
Intuitively, $a$ is least divisive and $b$ and $c$ are equally divisive. Yet, most SCF-based DSFs will return the full set $\{a,b,c\}$, thereby declaring all three proposals equally divisive.
To see this, observe that, for any subelectorate, $a$ will always be the unanimously top-ranked proposal and thus would be selected by any SCF that respects such unanimous preferences (which includes most of the commonly studied voting rules).
\end{example}

\noindent
We note that, in a technical sense, Definition~\ref{def:scfDSF} is a special case of Definition~\ref{def:scoreDSF}, because every SCF $F$ can also be thought of as a scoring function $s$, namely the one awarding a score of $\nicefrac{1}{k}$ for a win as part of a group of winners of size~$k$ (and a score of $0$ otherwise).\footnote{Interestingly, if we were to generalise Definition~\ref{def:scfDSF} to allow for \emph{probabilistic} SCFs~\citep{BrandtTRENDS2017}, returning for any given profile of preferences a lottery over alternatives, then that more general definition would again encompass Definition~\ref{def:scoreDSF}, as we can think of scores as (being proportional to) probabilities for getting selected in such a lottery. We shall not pursue this path any further here but note that it might offer a principled way of dealing with the problem highlighted by Example~\ref{ex:fallacy} without giving up on the modelling flexibility provided by the score-based approach to defining a DSF.}
When comparing Definition~\ref{def:scoreDSF} and Definition~\ref{def:scfDSF}, the latter seems more ``principled'' (in the sense of not suffering from the issue identified in Example~\ref{ex:fallacy}), while the former has the advantage of being more ``fine-grained'' (thereby making it easier to avoid the issue raised in Example~\ref{ex:scf-unanimous-indecisive}).

\paragraph{Definitions based on profile indices.}
In the introductory section, we already mentioned some of the parallels between the task of defining the divisiveness of a proposal \emph{within} a profile and the task of classifying an entire profile as being, for instance, \emph{diverse} or \emph{polarised}. There is a growing literature on formalising such concepts using the toolkit of social choice theory \citep{AlcaldeUnzuVorsatzSCW2013,CanEtAlMSS2015,HashemiEndrissECAI2014,FaliszewskiEtAlIJCAI2023}. Given the seemingly obvious connections, can we refashion some of these concepts to arrive at an attractive definition for a DSF?

Let $\delta : \Profs \to \mathbb{R}$ be an index function mapping any given profile $R$ to a real number, which we can think of as reflecting the degree of diversity, polarisation, or disagreement in~$R$. All of the measures discussed in the aforementioned literature are of this form~\citep{AlcaldeUnzuVorsatzSCW2013,CanEtAlMSS2015,HashemiEndrissECAI2014,FaliszewskiEtAlIJCAI2023}. 
A good example, also showing up as an example in all of the cited works, is the function~$\delta$ mapping any given profile~$R$ to the average Kendall tau distance across all pairs of preferences in~$R$.

The fundamental idea motivating the definition of our next class of DSFs is that removing a fairly divisive proposal~$x$ from a profile~$R$ should result in a new profile $R'$ of fairly low diversity/polarisation/disagreement $\delta(R')$. 
A slight reformulation of this idea is this: Moving a fairly divisive proposal~$x$ to the top of everyone's preference in a profile~$R$ should result in a new profile $R'$ with fairly low $\delta(R')$. In the interest of space, we shall formalise and analyse only this second variant of the idea, but the reader should expect findings to turn out very similar for the first variant.\footnote{The same is true if we were to move $x$ to the bottom of everyone's preference (or to some other fixed position) before applying~$\delta$.}

\begin{definition}\label{def:deltaDSF}
The \textbf{profile-index-based DSF} $\Delta_\delta$ induced by profile index function $\delta : \Profs \to \mathbb{R}$ is defined as follows:
\begin{eqnarray*}
\Delta_\delta(R) & = & \argmin_{x\in X} \delta(R^{\hat{x}}),
\end{eqnarray*}
where $R^{\hat{x}}$ is the profile we obtain if we move proposal~$x$ to the top of every agent's preference ranking in profile~$R$.
\end{definition}

\noindent
Again, the idea is that in $R^{\hat{x}}$ proposal~$x$ will not in any way contribute to diversity, polarisation, or disagreement. So the lower $\delta(R^{\hat{x}})$, the more divisive $x$ must have been in the original profile~$R$.

\section{Normative Requirements}\label{sec:axioms}

At this point we have seen a number of different ideas for how to define a DSF, highlighting the size of the design space available when choosing a concrete DSF for a concrete application. What we still lack is guidance for how to choose. In social choice theory, the time-honoured approach to help the designer choose in a situation like this is the axiomatic method~\citep{ArrowEtAlHBSCW2002,Thomson2023}. So in this section we formulate a number of axioms for DSFs: formal properties of a certain normative appeal that any given DSF might satisfy or violate, and that thus can serve as criteria for preferring one DSF over another. While we identify a number of axiomatic principles that can meaningfully guide the search for an adequate DSF, it is important to stress that---here and in any other work using the axiomatic method---an axiom should never be misunderstood as an absolute requirement.

\paragraph{Basic axioms.}
We start with two standard axioms used throughout all of social choice theory that also are relevant here.
First, we usually would not want the results returned by a DSF depend on the identities of the agents reporting preferences. This property is commonly known as \emph{Anonymity}.

\begin{axiom}
We say that a DSF $\Delta$ satisfies \textbf{Anonymity} in case $\Delta(R)=\Delta(R\circ\sigma)$ for any profile $R : N \to {X!}$ and bijection $\sigma : \mathbb{N} \to \mathbb{N}$.
\end{axiom}

\noindent
Here the function composition $R\circ\sigma$ is the profile we obtain when we let agents in~$R$ swap their preferences according to~$\sigma$ (amongst themselves or with agents outside of~$N$).

The familiar counterpart to Anonymity is \emph{Neutrality}, encoding a dual symmetry requirement with respect to proposals. 
It says that, if we change the proposals occurring in a profile, then we must change the proposals in the resulting set of most divisive proposals accordingly. We encode this axiom with the help of permutations $\sigma : X \to X$ on proposals. For the purposes of our formal definition, we extend any such $\sigma$ from elements of $X$ to both subsets of $X$ and functions from some $N$ to linear orders on $X$ in the natural manner.

\begin{axiom}
We say that a DSF $\Delta$ satisfies \textbf{Neutrality} in case $\Delta(\sigma(R)) = \sigma(\Delta(R))$ for any profile $R$ and permutation $\sigma : X \to X$.
\end{axiom}

\noindent
\citet{ColleyEtAlIJCAI2023}---at least implicitly---propose a further basic axiom, which we shall call \emph{Uniformity}, by postulating that, under any reasonable measure of divisiveness, in a \emph{perfectly uniform profile}, where each of the $m!$ possible preference rankings occurs equally often, we must declare all proposals equally divisive. 

\begin{axiom}
We say that a DSF $\Delta$ satisfies \textbf{Uniformity} in case $\Delta(R^U)  = X$ for any perfectly uniform profile~$R^U$.
\end{axiom}

\noindent
\citet{ColleyEtAlIJCAI2023} point out that $\Delta_{\textit{nav}}^s$ (of Definition~\ref{def:navarrete}) and several of its variants satisfy Uniformity. 
In fact, as we shall see next, Uniformity is satisfied whenever both Anonymity and Neutrality are satisfied.

\begin{proposition}
Every DSF that satisfies both Anonymity and Neutrality must also satisfy Uniformity.
\end{proposition}

\begin{proof}
Let $\Delta$ be a DSF that is anonymous and neutral, and
consider any perfectly uniform profile $R^U : N \to {X!}$. 
By definition, $\Delta(R^U)$ cannot be empty, so w.l.o.g., suppose $a\in\Delta(R^U)$.
We are done if we can show, for an arbitrary proposal $b\in X$, that also $b\in\Delta(R^U)$.

Let $\sigma_X : X \to X$ be the permutation on proposals that swaps $a$ and $b$ (and nothing else). 
As $R^U$ is perfectly uniform, there must be a bijection $\sigma_N : \mathbb{N} \to \mathbb{N}$ defined on all agents of the universe~$\mathbb{N}$ that acts as a permutation on~$N$ such that $\sigma_X(R^U) = R^U \circ \sigma_N$. (In words: any profile we can reach by swapping two proposals we can also reach by permuting some agents.) By Neutrality, we get $b\in\Delta(\sigma_X(R^U))$ and thus $b\in \Delta(R^U\circ\sigma_N)$. So by Anonymity, we get $b\in \Delta(R^U)$, meaning we are done.
\end{proof}

\noindent
Thus, as essentially all intuitively reasonable choices for a DSF will satisfy both Anonymity and Neutrality,\footnote{In particular, any DSF~$\Delta$ belonging to one of the three classes defined in Section~\ref{sec:classes} will satisfy Anonymity and Neutrality as long as the parameters ($p$, $s$, $F$, $\delta$) inducing $\Delta$ satisfy the corresponding symmetry requirements.} we do not need to consider the axiom of Uniformity as a selection criterion any further.

\paragraph{Unanimity axioms.}
Another common staple of social choice theory are unanimity axioms. In the context of designing voting rules, such an axiom would say that unanimously held preferences should be respected. For us, instead, it means that unanimously held preferences should not be taken as contributing to divisiveness. We start with the most basic instantiation of this fundamental principle.

We say that a profile $R : N \to {X!}$ is \emph{unanimous} in case it maps every agent to the same fixed ranking over proposals (meaning that $R_i=R_j$ for all $i,j\in N$). In this situation, none of the proposals is particularly divisive. So any reasonable DSF should declare all proposals in $X$ as being equally and thus also ``most'' divisive. 

\begin{axiom}\label{ax:profile-unanimity}
We say that a DSF $\Delta$ satisfies \textbf{Profile Unanimity} in case $\Delta(R)=X$ for any unanimous profile~$R$.
\end{axiom}

\noindent
This also is a very weak axiom that any reasonable DSF will satisfy. So let us now turn to more demanding unanimity principles.

Recall that we say that proposal~$x$ occurs in position~$k$ in preference $R_i$ in case the number of proposals ranked strictly above~$x$ by $R_i$ is exactly $k-1$. Arguably, if a given proposal $x$ occurs in the same position throughout a profile, it is minimally divisive---or at the very least not most divisive. Indeed, for any such proposal~$x$ there is unanimous agreement on the rank of $x$ within the space of all proposals, which is not something one would expect to see for a divisive proposal.\footnote{But also note that unanimous agreement on the absolute rank of $x$ is not the same thing as unanimous agreement on how $x$ should ranked relative to each of the other proposals. The latter condition is also of interest and could be instrumentalised to define a weaker unanimity axiom.} Let us turn this idea into an axiom.

\begin{axiom}\label{ax:pos-unanimity}
We say that a DSF $\Delta$ satisfies \textbf{Position Unanimity} in case $x \notin \Delta(R)$ for any profile~$R$ where proposal~$x$ occurs in the same position throughout (unless $R$ is a unanimous profile).
\end{axiom}


\noindent
The exception accounting for unanimous profiles ensures we do not violate Axiom~\ref{ax:profile-unanimity}.
The following weakening of Axiom~\ref{ax:pos-unanimity} relaxes the need to rank fixed-position proposals \emph{strictly} above the rest.
   
\begin{axiom}\label{ax:weak pos-unanimity}
We say that a DSF $\Delta$ satisfies \textbf{Weak Position Unanimity} in case $x \notin \Delta(R)$ or $\Delta(R) = X$ for any profile~$R$ where proposal~$x$ occurs in the same position throughout.
\end{axiom}

\noindent
While \citet{ColleyEtAlIJCAI2023} do not discuss (Weak) Position Unanimity \emph{as an axiom}, they clearly share the intuition that upholding this principle is central to the notion of divisiveness. One of their results can be interpreted as showing that $\Delta_{\textit{nav}}^s$ with normalised Borda scoring~$s$ satisfies Weak Position Unanimity, which they list as an argument in favour of using it. 
We now generalise their result to a large class of score-based DSFs.
Let us call a scoring function~$s$ \emph{positional} if it can be defined in terms of a scoring vector $(s_1,\ldots,s_{m})$, with $x$ garnering an extra $s_k$ points whenever an agent in $C$ ranks it at position~$k$ in the partial profile $R{\restriction_C}$ at hand.\footnote{The positional scoring rules mentioned in Example~\ref{ex:fallacy} are voting rules defined in terms of such positional scoring functions.} 
And let us call such an $s$ \emph{normalised} if it can be obtained from a positional scoring function by dividing all scores by $|C|$. The normalised Borda scoring function is an example.

\begin{proposition}\label{prop:score-based-weak-pos-unanimity}
Every score-based DSF $\Delta^p_s$ induced by a normalised positional scoring function~$s$ satisfies Weak Position Unanimity.  
\end{proposition}

\begin{proof}[Proof (sketch)]
Note that the claim entails that this positive result does not depend on the probability distributions $p_{R,x}$ at all. 

If $x$ occurs at the same position throughout $R : N \to {X!}$, then for any normalised positional $s$ and any nonempty coalition $C \subsetneq N$ we get $s(R{\restriction_C},x) = s(R{\restriction_{\complement{C}}},x)$, and thus $\mathit{div}_s(R,x,C,\complement{C})=0$. So the divisiveness value we obtain for $x$ will be~$0$, meaning no other proposal can have a strictly lower divisiveness value.
\end{proof}

\noindent
\citet{ColleyEtAlIJCAI2023} provide an example for a non-positional scoring function~$s$ for which the divisiveness value of proposals in a fixed position is not~$0$.\footnote{The example given by \citet{ColleyEtAlIJCAI2023} involves a scoring function known as the \emph{asymmetric Copeland score} \citep{ZwickerHBCOMSOC2016}. We note that this example does not work if we use the \emph{symmetric Copeland score} \citep{ZwickerHBCOMSOC2016} instead. We also stress that a proposal getting a nonzero divisiveness value might still be the least divisive proposal, so such examples do not immediately demonstrate a failure of the normative principle at stake.} So the same argument would not go through.
Also, we cannot strengthen Weak Position Unanimity to Position Unanimity in Proposition~\ref{prop:score-based-weak-pos-unanimity}. The 2-agent profile with $a\pref_1 b\pref_1 c$ and $a\pref_2 c\pref_2 b$ provides a simple counterexample when $s$ is the \emph{Plurality scoring function} (giving 1 point to each top proposal and 0 to the rest  \citep{ZwickerHBCOMSOC2016}), as the two subprofiles ``look the same'' under~$s$.

Unfortunately, the guarantees of Proposition~\ref{prop:score-based-weak-pos-unanimity} do not extend to SCF-based DSFs, as the following example illustrates.

\begin{example}\label{ex:scf-pos-unanimity}
Consider this 3-agent profile of preferences:
\[\begin{array}{c}
a\pref_1 x\pref_1 y\pref_1 b\pref_1 c \\
b\pref_2 x\pref_2 y\pref_2 c\pref_2 a \\
c\pref_3 x\pref_3 y\pref_3 a\pref_3 b 
\end{array}\]
Observe that any DSF that declares $x$ as being ``uniquely most divisive'' would fail (Weak) Position Unanimity.

But now consider the SCF-based DSF $\Delta^p_F$, where $F$ is the Borda rule and $p$ is the array of uniform distributions. There are three partitions to consider. For partition $\{\{1\},\{2,3\}\}$, proposal $a$ wins for $\{1\}$  and $x$ for $\{2,3\}$ Also for the other two partitions, $x$ always wins for the 2-agent coalition and $b$ and $c$, respectively, win for the 1-agent coalition. So $x$ alone will be declared most divisive.
\end{example}

\paragraph{Reinforcement axioms.}
The classical axiom of \emph{Reinforcement} for voting rules says that when for two profiles $R$ and $R'$, reported by disjoint electorates, we select two sets of proposals with a nonempty intersection, then for profile $R\oplus R'$ we must select that very intersection. Reinforcement is a natural requirement for voting rules and it is central to the axiomatisation of the Borda rule and the positional scoring rules \citep{YoungJET1974,YoungSIAM1975}. But as the next example illustrates, in the context of measuring divisiveness it is normatively inadequate.

\begin{example}
Let $X=\{a,b,c,d\}$ and let $R$ be a profile where six agents report preferences, all ranking $a$ at the top and each reporting a different ranking on the remaining proposals~$\{b,c,d\}$. In profile~$R’$ another six agents all rank~$a$ at the bottom and also each report a different ranking for $\{b,c,d\}$. Then any intuitively adequate DSF will return $\{b,c,d\}$ for both profiles, as $a$ is clearly not divisive at all \emph{within the context of each profile}. But for the union $R\oplus R'$ of the two profiles, proposal~$a$ clearly is the most divisive proposal, as half of the agents rank it at the top and the other half at the bottom, so we would want our DSF to return $\{a\}$ rather than $\{b,c,d\}$.
\end{example}

\noindent
Next we define a very weak form of Reinforcement that, contrary to the classical axiom, has clear normative appeal when it comes to measuring divisiveness. We obtain it by restricting Reinforcement to the special case where one profile is perfectly uniform.\footnote{\citet{HardingSCW2022} introduces the same axiom under the name of \emph{Uniform Voter Addition Invariance} in her work on proxy selection in liquid democracy.}

\begin{axiom}
We say that a DSF $\Delta$ satisfies \textbf{Uniform Reinforcement} in case $\Delta(R) = \Delta(R\oplus R^U)$ for any profile $R : N \to {X!}$ and any perfectly uniform profile $R^U : N' \to {X!}$ with $N \cap N' = \emptyset$.
\end{axiom}

\noindent
In words, this axiom (essentially) says that if we add $m!$ preferences to a given profile~$R$, one for each possible way of ranking the $m$ proposals, then this should not affect our divisiveness judgments.\footnote{This way of paraphrasing the axiom covers the case where $R^U$ contains each possible preference exactly once. The full axiom also covers the case where $R^U$ includes $k$ copies of each possible preference.} 
To be clear, adding $R^U$ to a given profile~$R$ certainly will ``dilute'' any disagreements present in $R$, and we might, for instance, expect the \emph{absolute} degree of divisiveness of any given proposal to reduce. But the \emph{relative} ranking of proposals in terms of divisiveness (and thus the property of being or not being maximally divisive) clearly should not be affected by this kind of operation, affecting all proposals in the exact same way.
Uniform Reinforcement thus appears to be a very weak requirement that we would expect (and hope) any reasonable DSF will satisfy. But it turns out that the rank-variance DSF (of Definition~\ref{def:rankvar}) does not.

\begin{proposition}
The rank-variance DSF $\Delta_{\textit{var}}$ violates Uniform Reinforcement (for any $m\geq 3$).
\end{proposition}

\begin{proof}[Proof (sketch)]
We are done if we can construct a counterexample for any given $m\geq 3$. We do so for $m=5$, as it is especially instructive, but the same construction is possible for any $m\geq 3$.

Let $X = \{a,b,c,d,e\}$ and consider the two-agent profile $R=(abcde,badce)$. The rank-variance will (correctly) assess $e$ as being minimally divisive (with variance~$0$) and return all other proposals as maximally divisive: $\Delta_{\textit{var}}(R) = \{a,b,c,d\}$. In particular, both $\{a,b\}$ and $\{c,d\}$ are treated symmetrically, with the former occurring at positions~$1$ and~$2$ (with $\mu^R(a)=\mu^R(b)=1.5$), and the latter occurring at positions~$3$ and~$4$ (with $\mu^R(c)=\mu^R(d)=3.5$).

Let $R^U$ be the perfectly uniform $120$-agent profile where each possible ranking of the five proposals occurs exactly once. Now consider the union $R\oplus R^U$. In this profile the positions of $a$ and $b$ have higher variance than those of $c$ and $d$, because the mean position of each one of the five proposals will now be close to~$3$. Worse, $e$ now has maximal variance. (To see this, also without going through the tedious calculations, observe that ranking a proposal twice in position~$5$ is the most extreme thing you can do against the backdrop of the $120$ preferences pushing the mean close to~$3$.) Thus, $\Delta_{\textit{var}}(R\oplus R^U) = \{e\}$, in violation of Uniform Reinforcement.
\end{proof}

\noindent
This result provides a clear axiomatic justification for rejecting rank variance as a suitable measure of divisiveness.\footnote{Amartya Sen explains why variance is not a good tool to measure income inequality in somewhat similar terms~\citep[p.~27]{Sen1997}.}
For the DSF families defined in Section~\ref{sec:classes}, whether a given DSF satisfies Uniform Reinforcement will depend on the specific parameters used to define it. For some, such as $\Delta_\delta$ with $\delta$ measuring average Kendall tau distance between pairs of preferences, it is easy to see that they do.

\paragraph{Accounting for clones.}
Two proposals $x$ and $y$ are called \emph{clones} in the context of profile~$R$ if they appear adjacent to one another in each of the individual rankings in~$R$. This concept was first introduced into social choice theory to model normative requirements regarding the protection of voting rules against ``vote splitting'' between similar candidates \citep{TidemanSCW1987}. In the context of modelling divisiveness, we would want proposals that are clones of one another to be treated ``similarly'' (in a sense yet to be defined). Treating them exactly the same, so either declaring both or neither of them most divisive, however, seems too demanding. Indeed, one of the clones might be a ``slightly more extreme'' variant of the other. Instead, to account for the possibility of such exeptions while still encoding the basic principle of treating similarly ranked proposals similarly in terms of divisiveness, one might want to postulate that no other (non-clone) proposal can be ranked between the two in terms of divisiveness. 

\begin{axiom}\label{ax:cc}
We say that a DSF $\Delta$ satisfies \textbf{Clone Consistency} in case $\{x,y\}\subseteq\Delta(R)$ entails $\{x,x',y\}\subseteq\Delta(R)$ for any profile $R$ in which proposals $x$ and $x'$ are clones but proposals $x$ and $y$ are not.
\end{axiom}

\noindent
Thus, if $\Delta$ declares both $x$ and $y$ most divisive (and if they are not clones of one another), then $\Delta$ must also declare all of their clones as most divisive (including in particular $x$'s clone~$x'$).

We shall soon see that finding a DSF that satisfies Clone Consistency is not easy, but there nonetheless is a large class of DSFs that does.
Let us call a profile index function $\delta : \Profs \to \mathbb{R}$ \emph{neutral} if $\delta(R) = \delta(\sigma(R))$ for any permutation $\sigma : X \to X$, i.e., if it is invariant under swapping proposals within a profile~\citep{AlcaldeUnzuVorsatzSCW2013,HashemiEndrissECAI2014}.

\begin{proposition}\label{prop:delta-clone}
Every profile-index-based DSF $\Delta_\delta$ induced by a neutral profile index function~$\delta$ satisfies Clone Consistency.
\end{proposition}

\begin{proof}
Consider a profile $R$ with clones $x$ and $y$. Then we can obtain $R^{\hat{x}}$ from $R^{\hat{y}}$ by swapping $x$ and $y$, meaning that we must have $\delta(R^{\hat{x}}) = \delta(R^{\hat{y}})$ for any neutral~$\delta$. Hence, the DSF $\Delta_\delta$ induced by such a $\delta$, when applied to $R$, must return either both or neither of $x$ and $y$. So $\Delta_\delta$ satisfies (a strong form of) Clone Consistency.
\end{proof}

\noindent
So, while it is certainly debatable whether our specific rendering of the basic normative principle of treating similar objects similarly in the form of Axiom~\ref{ax:cc} is the most appropriate way of doing so, Proposition~\ref{prop:delta-clone} shows that it is satisfied by a large class of natural DSFs, which suggests that it is not an unreasonable property to demand.

\paragraph{Popularity and divisiveness.}
Finally, we propose two axioms that speak to two opposing interpretations of ``divisiveness'', depending on whether the popularity of the proposals in question should matter when assessing their divisiveness. One view is that divisiveness should be independent from popularity. A consequence of this view is that when we \emph{invert} all the preference rankings in a given profile (e.g., changing $a\pref b\pref c$ to $c\pref b\pref a$), then that should not affect our divisiveness judgments for that profile.\footnote{\citet{DelemazureEtAlIJCAI2024} refer to what we call \emph{Inversion Invariance} as \emph{Reverse Stability}.} 

\begin{axiom}
We say that a DSF $\Delta$ satisfies \textbf{Inversion Invariance} in case $\Delta(R)=\Delta(R')$ for any profile~$R$ and its inversion~$R'$.
\end{axiom}

\noindent
The other view is that the potential divisiveness of clearly unpopular proposals is not important, so a DSF should not return such proposals. This view is reasonable, for instance, if our ultimate goal is to choose a single proposal and we use a DSF to select a contender that requires further debate. A clear-cut example for a proposal that is unpopular is one that is Pareto-dominated by another proposal.

\begin{axiom}
We say that a DSF $\Delta$ satisfies \textbf{Pareto Efficiency} in case $x\notin\Delta(R)$ for any profile~$R$ and any Pareto-dominated proposal~$x$ in~$R$.
\end{axiom}

\noindent
Clearly, Inversion Invariance and Pareto Efficiency are in direct conflict with one another and no DSF can possibly satisfy both axioms.\footnote{To see this, consider the 1-agent profile $(a\pref b)$, where Inversion Invariance requires any DSF to return $\{a,b\}$, while Pareto Efficiency requires any DSF to return~$\{a\}$.} (This is not a negative result, but an intentional strategy to partition the space of possible definitions for a DSF.)

Let us state some simple results regarding these two axioms (proofs are omitted but immediate). 
First, any profile-index-based DSF $\Delta_\delta$ satisfies Inversion Invariance if $\delta$ satisfies the corresponding property.
Second, any SCF-based DSF \smash{$\Delta^p_F$} induced by a SCF~$F$ that is Pareto-efficient can violate Pareto Efficiency only in cases where it returns the full set~$X$.
Finally, for score-based DSFs \smash{$\Delta^p_s$} the situation depends on the scoring function~$s$. For instance, for Plurality scoring, we satisfy the same variant of Pareto Efficiency. 
But for (classical) Borda scoring, we satisfy Inversion Invariance (due to the Borda scoring vector's inherent symmetry).

\section{Impossibility Results}\label{sec:impossibilities}

In this section we present three impossibility results demonstrating that satisfying certain combinations of normatively desirable properties of a DSF turns out to be mathematically impossible. All three impossibilities involve the axiom of Position Unanimity, which due to its intuitive appeal and mathematical simplicity arguably is a central requirement for measuring divisiveness. These thus are significant negative results that highlight the difficulties of identifying an adequate measure of divisiveness. 

\paragraph{Pareto Efficiency.}
Our first impossibility result concerns a straightforward incompatibility of Position Unanimity with Pareto Efficiency, which applies even in the case of Weak Position Unanimity.

\begin{theorem}\label{thm:pareto}
No DSF can simultaneously satisfy Pareto Efficiency and Weak Position Unanimity (in case $m\geq 3$).
\end{theorem}

\begin{proof}
Consider a profile $R : N \to {X!}$ in which agents report preferences over a set $X = \{a,b_1,\ldots,b_{m-1}\}$ of $m\geq 3$ proposals. Suppose each of the following two rankings occurs at least once in~$R$ (and no other ranking shows up in the profile): 
\[\begin{array}{c}
a\pref b_1\pref \cdots \pref b_{m-1} \\
a\pref b_{m-1}\pref \cdots \pref b_1 
\end{array}\]
Any DSF that satisfied Pareto Efficiency must return $a$ as the (only) maximally divisive proposal.
On the one hand, Weak Position Unanimity requires that $a$ is not declared maximally divisive, so there indeed exists no DSF that satisfies both of our requirements.
\end{proof}

\noindent
Observe that, if we were to weaken Pareto Efficiency to permit the DSF to return the full set~$X$ even when there are dominated proposals, we still get an impossibility in case we at the same time strengthen Weak Position Unanimity to Position Unanimity. As we saw that all reasonable SCF-based DSFs satisfy this weak variant of Pareto Efficiency, we can infer that no such DSF can satisfy Position Unanimity (we got a taste of this already in Example~\ref{ex:scf-pos-unanimity}).

We stress again that Pareto Efficiency is an attractive axiom for divisiveness only in very specific contexts, where we want to identify proposals that are both divisive and popular. So Theorem~\ref{thm:pareto} might be interpreted as saying that the context-specific axiom of Pareto Efficiency is incompatible with the generalist axiom of Position Unanimity. A possible way out might be to develop new context-specific axioms in the spirit of Position Unanimity that only affect proposals that are not Pareto-dominated.

\paragraph{Profile-index-based measures.}
If we want to satisfy Position Unanimity, we must abandon Pareto Efficiency as a requirement. A large and potentially attractive class of DSFs that violate Pareto Efficiency (because they satisfy Inversion Invariance) are the DSFs defined in terms of a profile index function $\delta$ (see Definition~\ref{def:deltaDSF}). 
Recall that $\delta$ might be a measure of polarisation or diversity~\citep{AlcaldeUnzuVorsatzSCW2013,CanEtAlMSS2015,HashemiEndrissECAI2014,FaliszewskiEtAlIJCAI2023}. 
Unfortunately, we obtain an impossibility also for this class.

\begin{theorem}\label{thm:delta-impossible}
No profile-index-based DSF $\Delta_\delta$ induced by a neutral profile index function~$\delta$ satisfies Position Unanimity (in case $m\geq 3$).
\end{theorem}

\noindent
We shall soon see that Theorem~\ref{thm:delta-impossible} can be obtained as a corollary of two other results. But a direct proof might provide extra insight.

\begin{proof}[Proof of Theorem~\ref{thm:delta-impossible}]
Consider scenarios with an odd number $m\geq 3$ of proposals, $a_1$ to $a_m$, where the same number of agents each report one of the following two rankings (the proof for even numbers $m$ of proposals is similar):
\[\begin{array}{c}
a_1\pref \cdots \pref a_{\lceil\nicefrac{m}{2}\rceil} \pref \cdots \pref a_m \\
a_m\pref \cdots \pref a_{\lceil\nicefrac{m}{2}\rceil} \pref \cdots \pref a_1 
\end{array}\]
Position Unanimity requires that $a_{\lceil\nicefrac{m}{2}\rceil}$, the only proposal occurring in the same position throughout, is not returned by the DSF.

On the other hand, whichever proposal $a_j$ we move to the top, we end up with the same kind of new profile (modulo renaming of proposals): there is one proposal at the top everywhere (the one we moved), and the rest of the profile is ``perfectly polarised''. So any DSF $\Delta_\delta$ with a neutral index function~$\delta$ must treat all proposals the same. But this is in conflict with the requirement, imposed by Position Unanimity, of giving special treatment to $a_{\lceil\nicefrac{m}{2}\rceil}$.
\end{proof}

\noindent
At the conceptual level, a tentative explanation for this impossibility, and more specifically the problem observed in the proof, might be that even though $a_{\lceil\nicefrac{m}{2}\rceil}$ intuitively is less divisive than the other proposals, \emph{it still contributes to} diversity (or polarisation or disagreement) within the profile, by pushing other proposals further apart. 
The profile-index-based approach to measuring divisiveness is not able to distinguish between these two distinct concepts: being a divisive proposal and making other proposals look divisive. This also illustrates again the fact, now at a technical level, that, despite the obvious connections, the concept of divisiveness (of a single proposal) is really fundamentally different from the concepts of polarisation and diversity (of a profile of preferences over proposals) previously studied in the literature.

In principle, we can circumvent the impossibility of Theorem~\ref{thm:delta-impossible} by replacing Position Unanimity with Weak Position Unanimity. But this does not help in identifying an appealing DSF. The only obvious example that comes to mind is the DSF induced by the \emph{constant} profile index function, mapping any given profile to the same fixed number, which of course is of no practical interest.

\paragraph{Clone Consistency.}
Our final result shows that Clone Consistency and Position Unanimity represent incompatible requirements, at least in the presence of weak and natural symmetry requirements. 

\begin{theorem}\label{thm:clone-impossible}
No DSF can simultaneously satisfy Anonymity, Neutrality, Clone Consistency, and Position Unanimity (in case $m\geq 3$). 
\end{theorem}

\begin{proof}
Suppose we were to attempt to design a DSF $\Delta$ that has the required properties.
Consider the profile $R = (abc,cba)$, which has $b$ occurring in the same middle position throughout.
By Anonymity and Neutrality, either both or neither of $a$ and $c$ must be part of $\Delta(R)$.
Furthermore, as $R$ is not a unanimous profile, due to Position Unanimity, we must have $b \notin \Delta(R)$. Thus, as $\Delta(R)$ cannot be empty, we must have $\Delta(R)=\{a,c\}$.
But this choice is in direct conflict with Clone Consistency, as $a$ and $c$ are not clones, yet $a$ and $b$ are---so once $a$ and $c$ are declared maximally divisive, $b$ would have to be declared maximally divisive as well.
\end{proof}

\noindent
This is both a negative and a surprising result, given that each of the axioms involved not only is of some normative appeal but also appears to have fairly moderate logical strength when considered in isolation. We note that, 
also here, the impossibility does not persist if we replace Position Unanimity by Weak Position Unanimity, as in that case the trivial DSF always returning the full set~$X$ would constitute a counterexample---but this of course does not provide us with a helpful way in in which to measure divisiveness.
A more promising route might be to weaken Clone Consistency. Indeed, while it is hard to disagree with our postulate that ``similar proposals should be treated similarly'', one possible interpretation of Theorem~\ref{thm:clone-impossible} is that our formalisation of that postulate is too narrow (even if it is sufficiently weak to admit the positive result of Proposition~\ref{prop:delta-clone}).

We conclude this section by observing that, in view of Proposition~\ref{prop:delta-clone}, we also can prove Theorem~\ref{thm:delta-impossible} as a corollary of Theorem~\ref{thm:clone-impossible}.

\section{Conclusion}\label{sec:conclusion}

We argued that being able to make formally sound judgments on the divisiveness of proposals for which we have collected preferences has important applications in digital democracy.
We therefore introduced a formal model of divisiveness, proposed three general approaches for retrieving maximally divisive proposals, and formulated axioms encoding relevant normative requirements. Our technical results highlight some of the fundamental difficulties associated with settling on a satisfactory definition of divisiveness. But we stress that, just as for classical impossibility results in social choice theory, these limitative results should not be interpreted as indicating that all attempts to use measures of divisiveness to analyse preference data are bound to fail, but rather that much more research on this broad topic is needed. 

Investigating a fourth approach to defining a DSF, only briefly mentioned earlier, namely to base its definition on a suitable probabilisitic SCF \citep{BrandtTRENDS2017} bears some promise (without, of course, affecting those of our impossibility results that do not only concern a specific class of DSFs). Such an approach might take care of concerns about score-based definitions not being grounded in axioms, while allowing for more fine-grained judgments than the basic SCF-based approach.

As our interest here has been in normative questions, we did not discuss algorithmic matters at all (which should in any case come \emph{after} one has identified a normatively appealing solution). But it is clear that most interesting measures of divisiveness will be computationally demanding. For score-based and SCF-based measures much of the complexity will come from the need to cycle through all decompositions of the electorate. A pragmatic solution to this problem might be to \emph{sample} from the distribution over decompositions rather than to perform an exact computation.

Similarly, we did not discuss empirical questions---which were explored in depth by \citet{NavarreteEtAlNHB2024}. Instead, our purpose here has been to fill a gap in that work by developing axiomatic foundations for the preference-based measurement of divisiveness. We hope this will lay the foundations for proposals of new candidate definitions for measures of divisiveness that perform adequately from an axiomatic point of view, at which point they also should be tested empirically on a wide range of preference data.

Our focus has been on preferences conforming to the classical model of social choice, where everyone reports a strict ranking of all proposals. Due to its simplicity and pervasiveness in the literature, this is the right model to start with. But divisiveness should also be studied for other models, such as weak preferences \citep{ArrowEtAlHBSCW2002}, approval preferences \citep{LaslierSanverHBAV2010}, combinations of approval and ranking \citep{BramsSanver2009}, Likert scales, linguistic grades \citep{BalinskiLarakiPNAS2007}, and incomplete preferences~\citep{KonczakLangMPREF2005,TerzopoulouPhD2021}.


\paragraph{Acknowledgements.}
I would like to thank Th\'eo Delemazure, Umberto Grandi, C\'esar Hidalgo, audience members at the COMSOC Seminar at the University of Amsterdam and the Seminar in Microeconomics at the University of Lausanne, as well as several anonymous reviewers for helpful feedback.
Funded by the European Union. Views and opinions expressed are however those of the author only and do not necessarily reflect those of the European Union or the European Research Council Executive Agency. Neither the European Union nor the granting authority can be held responsible for them. This work is supported by ERC Synergy Grant ADDI (\href{https://doi.org/10.3030/101166894}{doi.org/10.3030/101166894}).


\bibliography{div}

\end{document}